\documentclass[letterpaper, 10 pt, conference]{ieeeconf}  

\IEEEoverridecommandlockouts                   

\usepackage{graphics}
\usepackage{epsfig} 
\usepackage{mathptmx}
\usepackage{times}
\usepackage{amsmath}
\usepackage{amssymb}
\usepackage{dsfont}
\usepackage{enumerate}
\usepackage{multirow}
\usepackage{bm}
\usepackage[noadjust]{cite}
\usepackage{hyperref}

\usepackage{color}

\title{SIS epidemics on open networks: A replacement-based approximation
}

\author{Renato Vizuete, Paolo Frasca, and Elena Panteley
\thanks{This work was supported by F.R.S.-FNRS via the \emph{KORNET} project and via the Incentive Grant for Scientific Research (MIS) \emph{Learning from Pairwise Comparisons}, and by the \emph{RevealFlight} Concerted Research Action (ARC) of the Fédération Wallonie-Bruxelles and in part by the Agence Nationale de la Recherche (ANR) via grants HANDY ANR-18-CE40-0010 and COCOON ANR-22-CE48-0011.}
\thanks{R.~Vizuete is with ICTEAM institute, UCLouvain, B-1348, Louvain-la-Neuve, Belgium. R.~Vizuete is a FNRS Postdoctoral Researcher - CR. 
P.~Frasca is with Univ.\ Grenoble Alpes, CNRS, Inria, Grenoble INP, GIPSA-lab, F-38000 Grenoble, France. 
E.~Panteley is with 
Université Paris-Saclay, CNRS, CentraleSupélec, Laboratoire 
des signaux et systèmes, 91190, Gif-sur-Yvette, France.
{\tt\small renato.vizueteharo@uclouvain.be},
{\tt\small paolo.frasca@gipsa-lab.fr},\protect\\
{\tt\small elena.panteley@l2s.centralesupelec.fr}.}
} 

\newcommand{\RVnew}[1]{{\color{black}#1}}

\newcommand{\Gcal}{\mathcal G}
\newcommand{\Vcal}{\mathcal V}
\newcommand{\Ecal}{\mathcal E}

\newcommand{\vertiii}[1]{{\left\vert\kern-0.25ex\left\vert\kern-0.25ex\left\vert #1 
    \right\vert\kern-0.25ex\right\vert\kern-0.25ex\right\vert}}

\newcommand{\Var}{\, \mathrm{Var}  }    

\newtheorem{definition}{Definition}
\newtheorem{assumption}{Assumption}
\newtheorem{theorem}{Theorem}

\newtheorem{proposition}{Proposition}
\newtheorem{lemma}{Lemma}
\newtheorem{remark}{Remark}

\newcommand{\prt}[1]{\left(#1\right)}
\newcommand{\brk}[1]{\left[#1\right]}
\newcommand{\brc}[1]{\left\{#1\right\}}
\newcommand{\abs}[1]{\left|#1\right|}
\newcommand{\norm}[1]{\abs{\abs{#1}}}

\newcommand{\inProd}[2]{\langle#1,#2\rangle}

\newcommand{\R}{\mathbb R}

\newcommand{\Ep}[1]{\mathbb{E}\left[#1\right]}

\begin{document}

\maketitle
\thispagestyle{empty}

\begin{abstract}
In this paper we analyze continuous-time SIS epidemics subject to arrivals and departures of agents, by using an approximated process based on  replacements. In defining the SIS dynamics in an open network, we consider a stochastic setting in which arrivals and departures take place according to Poisson processes with similar rates, and the new value of the infection probability of an arriving agent is drawn from a continuous distribution. Since the system size changes with time, we define an approximated process, in which replacements take place instead of arrivals and departures, and we focus on the evolution of an aggregate measure of the level of infection. So long as the reproduction number is less than one, the long-term behavior of this function measures the impact of the changes of the set of agents in the epidemic. We derive upper bounds for the expectation and variance of this function and we include a numerical example to show that the approximated process is close to the original SIS process.
\end{abstract}

\section{Introduction}

Epidemic models are dynamical systems used to model contagion in a community of individuals. These types of models have long time been a key class of network dynamics and their interest has further grown in the last few years, specially due to the COVID-19 outbreak that impacted the entire world. 
One of the most important epidemic models is the SIS (Susceptible-Infected-Susceptible) compartmental epidemic model, which describes diseases where the agents, after being recovered, have no immunity against \linebreak
re-infection \cite{hethcote2000mathematics,kiss2017mathematics}. 

Most of the stability results derived using network models depend on the spectral radius of the adjacency matrix associated with the network, assuming that the composition of the system is fixed \cite{nowzari2016epidemics}. Nevertheless, in more realistic environments, individuals do not stay in a fixed place for a long time and are constantly moving between different locations. It is clear that mobility is one of the most important factors that need to be considered in the analysis of modern spreads of diseases \cite{alamo2022challenges}, and many works begin to consider this aspect to obtain more accurate results \cite{tizzoni2014use,hazarie2021interplay}. At the level of network models, time-varying networks have traditionally been used to model to some degree the mobility of agents by changing the interactions among them \cite{ogura2016stability,pare2018epidemic,leitch2019toward,zino2021analysis}. However, this type of models might  fail to represent modern networks where due to the development of transportation infrastructures, agents are constantly entering and leaving towns, cities and countries, which implies not only changes of the connections but also modifications of the infection probabilities associated with the agents. Other approaches have considered patchy environments where individuals move between layers and participate in the local dynamics \cite{abhishek2023sis}. Nevertheless, a complete knowledge of all the agents in the system and the transitions between populations might demand a lot of information, which limits the applicability of this approach.

The analysis of systems characterized by 
a time-varying set of agents is the object of study of \emph{open multi-agent systems}, where the dimension of the system may vary with time \cite{hendrickx2017open,franceschelli2021stability,vizuete2024trends}. At the level of large populations of individuals (e.g., countries, cities), it is common to have similar rates of arrivals and departures such that the variations of the size of the populations are almost negligible. However, even if the total number of individuals in a specific place could remain almost invariant, it is unrealistic to assume that exactly the same agents stay in a fixed location during all the evolution of the epidemic.
If we consider a similar rate of departures and arrivals, which preserve the number of individuals, epidemics under these conditions can be approximated as a dynamical system subject to replacements of individuals: a similar approach has been taken for opinion dynamics by \cite{carletti2008birth,torok2013opinions}. 

In this paper, we study a continuous time formulation of an SIS epidemic over a network of individuals when arrivals and departures of agents take place. We consider a stochastic setting where the time instants at which the events occur are determined by a homogeneous Poisson process and the infection probabilities of arrival agents are drawn from a continuous distribution.
Because the dimension of the system changes with time, the analysis of the epidemic is performed by using an aggregate function (i.e., scalar value) that measures the level of infection.
Due to the high complexity of the stochastic process, when the rates of arrivals and departures are similar, we propose an approximated version of the original process based on replacements. 
Upper bounds for the asymptotic values of the expectation and variance of the aggregate function are derived as functions of the parameters of the epidemic process.

\section{SIS model on closed networks}
We begin by recalling, in this section, the definition and main features of the classical SIS model on ``closed'' networks (i.e., networks whose node set does not change).

We consider a system composed by $n$ agents interacting through a connected undirected graph $\Gcal = (\Vcal,\Ecal)$ where the set of agents is given by $\Vcal=\{1,\ldots,n\}$ and the set of edges by $\Ecal\subseteq \Vcal\times \Vcal$. 
The SIS epidemic model over a network is given by \cite{nowzari2016epidemics}:
$$
\dot x_i(t)=-\delta x_i(t)+\sum_{j=1}^n a_{ij}\beta x_j(t)\prt{1-x_i(t)},
$$
where $x_i(t)\in[0,1]$ is the probability of infection of an individual at time $t$, $\delta$ is the recovery rate, $\beta$ is the infection rate, $a_{ji}=a_{ij}>0$ if there is an interaction between agents $i$ and $j$, and $a_{ij}=0$ otherwise.

The model of all the network can be expressed as:
\begin{equation}\label{eq:epidemics_network}
\dot x(t) = \prt{\beta A-\delta I}x(t)-\beta X(t)Ax(t),    
\end{equation}
where $I$ is the identity matrix, $A=[a_{ij}]$ is the adjacency matrix of the network, and $X(t)=\mathrm{diag}[x_1(t),\ldots,x_n(t)]$. Since the graph is undirected, $A$ is symmetric and its eigenvalues are denoted, in non-increasing order of magnitude as \linebreak $\lambda_1\ge\cdots\ge\lambda_n$.

\begin{proposition}[Stability \cite{lajmanovich1976deterministic}]\label{prop:stability_epidemics}
For the SIS epidemic \eqref{eq:epidemics_network}, the disease-free equilibrium $x=0$ is globally asymptotically stable if and only if
\begin{equation}\label{eq:stability_condition_graph}
\lambda_1\frac{\beta}{\delta}<1.    
\end{equation}
\end{proposition}

Since in open multi-agent systems the size of the system may change in time due to decoupled arrivals and departures \cite{hendrickx2017open,franceschelli2021stability,vizuete2024trends}, it is desirable to perform the analysis with the use of an aggregate function $V(x)$ that condensates the associated process in a scalar value, independently of the size of the system. In the case of SIS epidemics, a natural choice is the Lyapunov function $V:[0,1]^n\to\R_{\ge0}$:
\begin{equation}\label{eq:aggregate_function_norm}
V(x)=\frac{1}{n}\norm{x}^2, 
\end{equation}
where $\norm{\cdot}$ is the Euclidean norm. 
This Lyapunov function satisfies
\begin{equation}\label{eq:limit_V}
\lim_{t\to\infty}V(x(t))=0,
\end{equation}
when the epidemic is stable. 
Thanks to normalizing by the number of agents, $V(x)$ enjoys uniform bounds, independent of the size of the system:
\begin{equation}\label{eq:bounded_aggregate_function}
0\le V(x)\le 1.    
\end{equation}

In a closed system, where the evolution of the states of the agents is fully determined by   \eqref{eq:epidemics_network}, the dynamics of the aggregate function $V(x(t))$ can be upper bounded as follows (see \cite{khanafer2014infection}):
\begin{align}
    \dot{V}(x(t))&=\frac{2}{n}x^T(t)\dot x(t)\nonumber\\
    &=\frac{2}{n}x^T(t)\prt{\beta A-\delta I}x(t)-\beta x^T(t)X(t)Ax(t)\label{eq:original_dynamics_V}\\
    &\le \frac{2}{n}x^T(t)\prt{\beta A-\delta I}x(t)\nonumber\\
    &\le 2(\beta\lambda_1-\delta)V(x(t))\label{eq:mapping_f_norm},
\end{align}
where we used the fact that $x^T(t)X(t)Ax(t)\ge 0$ since $x_i\in[0,1]$. Clearly, the right-hand side of \eqref{eq:mapping_f_norm} is negative if
$\beta\lambda_1-\delta<0$,
which corresponds to the stability condition \eqref{eq:stability_condition_graph}. Finally, by applying the Comparison Lemma, we have that $V(x(t))$ satisfies:
\begin{equation}\label{eq:SIS_continuous_dynamics}
V(x(t))\le V_0+ 2(\beta\lambda_1-\delta) \int_0^t  V(x(\tau)) d\tau,
\end{equation}
where $V_0$ denotes the value of $V(x(t))$ at the time $t=0$. Similar bounds, in differential and integral form, will be derived for open SIS epidemics in the rest of this paper.

\section{\RVnew{Open SIS epidemics}}\label{sec:replacements_SIS}

In this section, we formulate a SIS epidemic where in addition to the continuous evolution of the states of the individuals according to \eqref{eq:epidemics_network}, the system also experiences arrivals and departures of individuals. Then, in the next section we will propose an alternative process, characterized by replacements in a fixed network, as an approximation in the case of similar rates of arrivals and departures.

\begin{definition}[Departure]
We say that an agent $j\in\Vcal(t^-)$ leaves the epidemic at time $t$ if:
$$
\Vcal(t^+)=\Vcal(t^-)\setminus \{j\},
$$
where $\Vcal(t^-)$ is the set of agents before the departure and $\Vcal(t^+)$ is the set of agents after the departure of agent $j$. Thus, $\abs{\Vcal(t^+)}=\abs{\Vcal(t^-)}-1$.
\end{definition}

\begin{assumption}[Departure process]\label{ass:departure_poisson_process}
The departure \! instants in the system are determined by a Poisson process $N_t^{(d)}$ with rate $\mu_d>0$ for $n>1$ and $\mu_d=0$ for $n=1$. During a departure, an agent $j$ in the current epidemic is chosen uniformly and removed from the system.
\end{assumption}

Notice that if there is only one agent in the epidemic, we keep this agent and we set to 0 the rate of the Poisson process for departures such that only arrivals can happen. In this way, the departure process becomes a slightly truncated version of a homogeneous Poisson process.

\begin{definition}[Arrival]
We say that an agent $j$ joins the epidemic at time $t$ if:
$$
\Vcal(t^+)=\Vcal(t^-)\cup \{j\},
$$
where $\Vcal(t^-)$ is the set of agents before the arrival and $\Vcal(t^+)$ is the set of agents after the arrival of agent $j$. Thus, $\abs{\Vcal(t^+)}=\abs{\Vcal(t^-)}+1$.
\end{definition}

\begin{assumption}[Arrival process]\label{ass:arrival_poisson_process}
The arrival instants are determined by a homogeneous Poisson process $N_t^{(a)}$ with rate $\mu_a>0$. 
During an arrival, an agent $j$ joins the epidemic with an infection probability $x_{j}$ determined by a random variable $\Theta$, which takes values according to a continuous distribution with  support in the interval $[0,1]$, with mean $m$ and variance $\sigma^2$. The new agent generates edges with each current node with probability $p$.
\end{assumption}

Notice that at any time the graphs are realizations of an Erd\H{o}s-Rényi graph and therefore $\lambda_1$ grows with the number of agents.
We consider that the recovery rate remains constant $\delta_n=\bar \delta$ and the infection rate is given by $\beta_n=n^{-1}\bar\beta$, which means that, as the graph grows in size, the healing rate (which depends on each individual) remains constant, whereas the infection rate decreases. This natural scaling law is also chosen
in \cite{vizuete2020graphon,gao2019spectral}. Indeed, on dense graphs this assumption means that in larger graphs,
even though there are more potential interactions, the average strength of the connections is suitably adjusted: this fact accounts for natural limitations in the rates
of contact between individuals. 

Similarly to \cite{vizuete2020graphon}, the objective of the aggregate function $V$ is to quantify the reaction of the epidemic to arrivals and departures that affect the evolution of the disease, such that any deviation from the equilibrium $x=0$ is due to the openness of the system.
The evolution of $V$ in this open setting is given by a stochastic differential equation (SDE) of the form:
\begin{equation}\label{eq:SDE_arrivals_departures}
dV=f(V(t))dt+g_a(V(t))dN_t^{(a)}+g_d(V(t))dN_t^{(d)},    
\end{equation}
where $f$ determines the dynamics of $V$ in continuous time, $g_a$ determines the change of $V$ during an arrival and $g_d$ determines the change of $V$ during a departure. This SDE is also known as a \emph{Poisson Counter driven Stochastic Differential Equation} \cite{brockett2009stochastic,brockett1999queueing}, and its solution is a cadlag function associated with a Piecewise-deterministic Markov Process (PDMP) such that it can also be studied using the theory of PDMP \cite{jacobsen2006point}.

Each realization of the stochastic process \eqref{eq:SDE_arrivals_departures} can be completely different, and the analysis of the asymptotic behavior is not appropriate because the limit $\lim_{t\to\infty}V(x(t))$ does not exist due to the continuous arrivals and departures. For this reason, we are interested in the statistical properties of \eqref{eq:SDE_arrivals_departures}, corresponding to $\Ep{V(x(t))}$. 
However, notice that the stochastic process \eqref{eq:SDE_arrivals_departures} is highly complex since the mapping $f$ determined by \eqref{eq:original_dynamics_V} is time-varying due to $A(t)$, and depends on the two Poisson process $N_t^{(a)}$ and $N_t^{(d)}$. In addition, due to the slightly truncation at $n=1$, the Poisson process $N_t^{(d)}$ is not homogeneous.

For this reason, we look for a more tractable mathematical process that approximates the behavior of the original process \eqref{eq:SDE_arrivals_departures}.

\section{Replacements as an approximation}

Due to the complexity of the stochastic process, we focus on the particular case of similar rates of arrivals and departures and we perform an analysis of the evolution of $V$ considering the following two approximations: 

\paragraph{Replacements}
In large populations, the rates of arrivals and departures are similar such that even if the individuals are constantly changing, the total number of agents remains the same in expectation. For this reason, we approximate arrivals and departures as replacements. 

\paragraph{Fixed network}
The connections of all the agents are given by an Erd\H{o}s-Rényi graph. This means that all the graph topologies generated in the stochastic process during arrivals have the same expectation $\bar A=p(\mathds{1}_n\mathds{1}^T_n-I_n)$, where $\mathds{1}_n$ is a vector of ones of size $n$. In addition, in large networks, the eigenvalues of the expected graph are close to those of the random graphs~\cite{chung2011spectra}. For this reason, we approximate the time-varying matrices $A(t)$ with a fixed matrix $\bar A$ of size $n_0$, where $n_0$ denotes the initial number of agents $n(t)$ at  $t=0$. 

These approximations are made formal by the following definitions.

\begin{definition}[Replacement]
We say that an agent $j\in\brc{1,\ldots,n_0}$ is replaced at time $t$ if:
$$
x^+=[x_1^-,\cdots,x_{j-1}^-,\RVnew{b_j},x_{j+1}^-,\cdots,x_n^-]^T,\qquad \RVnew{x_j^{-}\neq b_j},
$$
where $x^-=x(t^-)$ is the state of the system before the replacement, $x^+=x(t^+)$ is the state of the system after the replacement and \RVnew{$b_j$ is the probability of infection of the new agent.}
\end{definition}

\begin{assumption}\label{ass:invariant_connections}
The set of edges $\Ecal$ of the network remains invariant for all time $t$.
\end{assumption}

In the following proposition we analyze the variation of the aggregate function $V(x(t))$
 under a replacement event where the replaced agent is chosen uniformly (i.e., the probability of replacement of each agent is the same) and the value of the new agent after the replacement is determined by a \RVnew{random variable $\Theta$ following a} continuous distribution with support in the interval $[0,1]$, with mean $m$ and variance $\sigma^2$.

\begin{proposition}[Replacement]\label{prop:E_V_jumps}
During the replacement of an agent, the aggregate function $V(x(t))$ defined in \eqref{eq:aggregate_function_norm} satisfies:
\begin{equation}\label{eq:V_plus_epidemics}
\Ep{V(x(t^+))}\!-\!\Ep{V(x(t^-))}\!=\!\frac{1}{n_0} \prt{\sigma^2+m^2-\Ep{V(x(t^-))}}.    
\end{equation}
\end{proposition}
\begin{proof}
Due to the space constraints, we avoid including the time dependence of $x(t)$ and the state dependence of $V(x(t))$. This abuse of notation will be made multiple times in the rest of the paper.
Let us assume that an agent $i$ is replaced. We begin by computing the conditional expectation of $V^+=V(x(t^+))$ given $x^-$ and the value of the replaced agent $\Theta$.  Then, it holds
\begin{align*}    \Ep{V^+|x^-,\Theta}&=\frac{1}{n_0}\sum_{i=1}^{n_0}\prt{\frac{1}{n_0}\sum_{j\neq i}(x_j^2)^-+\frac{1}{n_0}\Theta^2}\\
    &=\frac{1}{n_0^2}\sum_{i=1}^{n_0}\prt{\norm{x^-}^2-(x_i^2)^-+\Theta^2}\\
    &=\frac{1}{n_0}\norm{x^-}^2-\frac{1}{n_0^2}\norm{x^-}^2+\frac{\Theta^2}{n_0}\\
    &=V^--\frac{1}{n_0}V^-+\frac{\Theta^2}{n_0},
\end{align*}
where $V^-$ denotes $V(x(t^-))$.
Then, by computing the total expectation, we complete the proof.
\end{proof}

\begin{remark}[Other systems]
The proof of Proposition~\ref{prop:E_V_jumps} did not use equation \eqref{eq:epidemics_network}, only the definition of $V$. Therefore, the statement remains valid for other systems characterized by replacements of agents.
\end{remark}

Along this work, we make the following assumption about the occurrence of replacements in the evolution of the epidemic.

\begin{assumption}[Replacement process]\label{ass:replacement_poisson_process}
The replacements instants are determined by a homogeneous Poisson process with rate $\mu>0$. 
During a replacement, an agent $j\in\brc{1,\ldots,n_0}$ is chosen uniformly and it is assigned a new infection probability \RVnew{$b_j$} determined by a random variable $\Theta$, which takes values according to a continuous distribution with  support in the interval $[0,1]$, with mean $m$ and \linebreak variance $\sigma^2$.    
\end{assumption}

\subsection{Pure replacements ($\dot x(t)=0$)}

Since the replacements process is independent of the evolution of the epidemic given by \eqref{eq:epidemics_network}, we analyze the behavior of the aggregate function $V$ during replacements independently of the continuous time dynamics by considering $\dot x(t)=0$.  
Notice that the jumps of the state $x(t)$ corresponds to the jump of the aggregate function $V(x(t))$. In this case, the value of $V(x(t))$ is given by:
\begin{align*}
    V(x(t))
    &=V_0+\sum_{k=1}^{N_t} \Delta V_k
    =V_0+\sum_{k=1}^{N_t}\prt{V(x(T_k^+))-V(x(T_k^-))},
\end{align*}
where $\Delta V_k$ is the size of the jump of the function $V(x(t))$ at the $T_k$ jump time of the Poisson process $N_t$ and
$$
V(x(T_k^-))=\lim_{t\uparrow T_k} V(x(t)) \qquad\!\!\!\text{and}\qquad\!\!\! V(x(T_k^+))=\lim_{t\downarrow T_k} V(x(t)).
$$
The aggregate function satisfies $V(x(t^-))=V(x(t^+))$ for almost all $t$, except in a countable number of time jumps. Therefore, it can be rewritten as~\cite{privault2013stochastic}:
\begin{align*}
    V(x(t))&=V_0+\sum_{k=1}^{N_t} \prt{V(x(T_k^+))-V(x(T_k^-))}\\
    &=V_0+\int_0^t(V(x(\tau^+))-V(x(\tau^-)))dN_{\tau}.
\end{align*}
A homogeneous Poisson process satisfies $\Ep{N_t}=\mu t$, which gives us $\Ep{dN_t}=d\Ep{N_t}=\mu dt$ \cite{privault2013stochastic}. Then, it holds
$$
\Ep{V(x(t))}=V_0+\int_0^t\Ep{V(x(\tau^+))-V(x(\tau^-))}\mu d\tau,
$$
which \RVnew{is the integral version of the ODE:}
\begin{equation}\label{eq:ODE_jump_system}
\frac{d}{dt}\Ep{V(x(t))}=\mu\Ep{ V(x(t^+))-V(x(t^-))}.
\end{equation}

\begin{proposition}[Pure replacement process]\label{prop:aggregate_function_replacements}
Under Assumption~\ref{ass:replacement_poisson_process} and inter-event dynamics $\dot x=0$, the aggregate function $V$ defined in \eqref{eq:aggregate_function_norm} satisfies:
\begin{equation}\label{eq:pure_jumps}
   \Ep{V(x(t))}=\prt{V_0-\sigma^2-m^2}e^{-\frac{\mu}{n_0}t}+(\sigma^2+m^2). 
\end{equation}
\end{proposition}

\vspace{1mm}

\begin{proof}
The result is obtained by applying \eqref{eq:V_plus_epidemics} in \eqref{eq:ODE_jump_system}
$$
\frac{d}{dt}\Ep{V(x(t))}=-\frac{\mu}{n_0}\Ep{V(x(t))}+\frac{\mu}{n_0}(\sigma^2+m^2),
$$
and solving the ordinary differential equation (ODE).
\end{proof}

Proposition~\ref{prop:aggregate_function_replacements} shows that in a system subject only to replacements, the expected value of the aggregate function is bounded and its asymptotic value is:
\begin{equation}\label{eq:asymptotic_aggregate_only_replacements}
\lim_{t\to\infty}\Ep{V(x(t))}=\sigma^2+m^2,    
\end{equation}
which corresponds to the second moment of the distribution used to generate the new values of the replaced agents. The rate of the Poisson process $\mu$ only has an influence on the rate of convergence of $V(x(t))$ such that a large value of $\mu$ will guarantee a fast convergence. The asymptotic value of $V(x(t))$ is independent of the rate $\mu$.

\section{\RVnew{First moment}}

Now, we consider the behavior of the aggregate function in continuous time subject to replacements at time instants determined by a Poisson process. Since the solution to \eqref{eq:original_dynamics_V} is unique and the jumps are determined by a homogeneous Poisson process, the aggregate function $V(x(t))$ is given \linebreak by \cite{brockett2009stochastic}:
\begin{multline}\label{eq:integral_SDE_V}
   V(x(t))=V_0+\int_0^t \frac{2}{n_0}\big(x^T(\tau)(\beta \bar A\\-\delta I)x(\tau)-\beta x^T(\tau)X(\tau)\bar Ax(\tau)\big) d\tau + \int_0^t \Delta V_k dN_\tau,     
\end{multline}
which can be expressed as the SDE
\begin{multline}\label{eq:SDE_V}
\!\!\!\!\!\!   dV(x(t))= \frac{2}{n_0}\prt{x^T(t)\prt{\beta \bar A-\delta I}x(t)-\beta x^T(t)X(t)\bar Ax(t)} dt \\+  \Delta V_k dN_t. 
\end{multline}
Notice that by denoting $$\omega(x(t))=\frac{2}{n_0}\prt{x^T(t)\prt{\beta \bar A-\delta I}x(t)-\beta x^T(t)X(t)\bar Ax(t)},$$ \eqref{eq:integral_SDE_V} can be expressed as
$$
V(x(t))\!=\!V_0+\!\!{\int_0^{T_1}\!\!\omega(x(\tau)) d\tau+\!\!\int_{T_1}^{T_2}\!\!\omega(x(\tau)) d\tau+\cdots}+\!\sum_{k=1}^{N_t} \!\Delta V_k,    
$$
which shows that $V(x(t))$ jumps at the time instants $T_k$ determined by the Poisson process, corresponding to the sum, and between these time intervals, $V(x(t))$ evolves according to \eqref{eq:original_dynamics_V}, corresponding to the integrals.

\begin{theorem}[Expectation]\label{thm:limits_expectation_SIS_fixed}
Consider a SIS epidemic satisfying the stability condition \eqref{eq:stability_condition_graph}. Under Assumptions~\ref{ass:invariant_connections} and \ref{ass:replacement_poisson_process}, the aggregate function $V(x(t))$ defined in \eqref{eq:aggregate_function_norm} satisfies:
\begin{equation}\label{eq:SIS_limit}
\limsup_{t\to\infty}\Ep{V(x(t))}\le\frac{\mu(\sigma^2+m^2)}{\mu+2n_0 \bar\delta-2 \bar\beta p(n_0-1)}.
\end{equation}
\end{theorem}

\vspace{1mm}

\begin{proof}
We take the expectation of both sides of \eqref{eq:SDE_V} and we obtain:
\begin{multline*}
    d\Ep{V}=\Ep{\frac{2}{n_0}\prt{x^T\prt{\beta \bar A-\delta I}x-\beta x^TX\bar Ax}} dt\\
   +\frac{\mu}{n_0} \prt{\sigma^2+m^2-\Ep{V}}dt, 
\end{multline*}
where we applied the Dominated Convergence Theorem to interchange the expectation and the derivative
$\Ep{d V}=d \Ep{V}$ and we used the fact that $\Ep{\Delta V}=\frac{1}{n_0} \prt{\sigma^2+m^2-\Ep{V}}$ according to Proposition~\ref{prop:E_V_jumps}. Then we obtain the following ODE:
\begin{align}
\dfrac{d}{dt}\Ep{V}&=\Ep{\frac{2}{n_0}\prt{x^T\prt{\beta \bar A-\delta I}x-\beta x^TX\bar Ax}}\nonumber\\ &\;\;\;\;+\frac{\mu}{n_0} \prt{\sigma^2+m^2-\Ep{V}}\nonumber\\
&\le 2(\beta\bar\lambda_1-\delta)\Ep{V}+\frac{\mu}{n_0} \prt{\sigma^2+m^2-\Ep{V}}\label{eq:expected_SDE}\\
&=\prt{\frac{2n_0 \beta\bar\lambda_1-2n_0 \delta-\mu}{n_0}} \Ep{V}+\frac{\mu(\sigma^2+m^2)}{n_0},\label{eq:SIS_ODE}    
\end{align}
where $\bar \lambda_1$ is the largest eigenvalue of $\bar A$.
By the Comparison Lemma, we can guarantee that $\Ep{V}$ is upper bounded by the solution of the right-hand side of \eqref{eq:SIS_ODE}. Then, it holds
\begin{align}
    \Ep{V}&\le \!e^{\prt{\frac{2n_0 \beta\bar\lambda_1-2n_0 \delta-\mu}{n_0}}t}\!\!\int_0^t \frac{\mu(\sigma^2\!+\!m^2)}{n_0}e^{-\prt{\frac{2n_0 \beta\bar\lambda_1-2n_0 \delta-\mu}{n_0}}\tau} d\tau\nonumber\\    
    &\quad +\Ep{V_0} e^{\prt{\frac{2n_0 \beta\bar\lambda_1-2n_0 \delta-\mu}{n_0}}t}\nonumber\\
    &=\prt{\Ep{V_0}-\frac{\mu(\sigma^2+m^2)}{\mu+2n_0 \delta-2n_0 \beta\bar\lambda_1}} \!e^{\prt{\frac{2n_0 \beta\bar\lambda_1-2n_0 \delta-\mu}{n_0}}t}\nonumber\\
    &\quad +\frac{\mu(\sigma^2+m^2)}{\mu+2n_0 \delta-2n_0 \beta\bar\lambda_1}\label{eq:solution_ODE_SIS}.
\end{align}
Finally, the result follows by taking the limit $t\to\infty$ in \eqref{eq:solution_ODE_SIS} and using $\bar\lambda_1=p(n_0-1)$, $\delta=\bar\delta$ and $\beta=n_0^{-1}\bar\beta$.
\end{proof}

Clearly, the bound \eqref{eq:SIS_limit} is a function of the parameters of the SIS epidemic (i.e., $\bar\beta$, $\bar\delta$ and $p$) and the parameters of the replacement process (i.e., $\mu$, $\sigma^2$ and $m^2$). Similarly to \eqref{eq:asymptotic_aggregate_only_replacements}, the bound is proportional to $\sigma^2+m^2$, but in this case the rate of the Poisson process $\mu$ also appears as a multiplication factor. The upper bound for the rate of convergence of the SIS epidemic in a closed system determined by \eqref{eq:mapping_f_norm} appears in the denominator of \eqref{eq:SIS_limit}, but this value is increased by the rate of the Poisson process $\mu$. From \eqref{eq:solution_ODE_SIS}, it can be seen that a large value of $\mu$ will increase the rate of convergence of $V(x(t))$, such that the convergence will be fast, but the asymptotic value will also increase. When $\mu\to\infty$, the bound \eqref{eq:SIS_limit} is given by $\sigma^2+m^2$, which coincides with the result of  Proposition~\ref{prop:aggregate_function_replacements} since the process will be characterized only by replacements.

Notice that the right-hand side in \eqref{eq:expected_SDE} corrresponds to the expectation of a SDE of the form:
\begin{equation}\label{eq:PCDSDE}
dV=f(V)dt+g(V)dN_t.    
\end{equation}
Unlike the SDE \eqref{eq:SDE_arrivals_departures} associated to the original process, the mapping $f$ is not time-varying and the Poisson process $N_t$ is homogeneous.

\section{Second moment}

Although the asymptotic behavior of the expected aggregate function $\Ep{V(x(t))}$ provides useful information of a system subject to the replacements of agents, it does not give any information about the deviation of single realizations from the expected value. To this purpose, it is necessary to analyze the variance of the aggregate function, denoted by Var$(V(x(t)))$. In this case, we use It\^o Lemma formulated for jump processes of the \linebreak form \eqref{eq:PCDSDE}. 

\begin{lemma}[It\^o Lemma \cite{brockett2009stochastic}]
For a continuous differentiable function $\phi:\R\to\R$:
$$
d \phi(V)=\inProd{\frac{d\phi}{dV}}{f(V)}dt+\brk{\phi(V+g(V))-\phi(V)}d N_t,
$$
where $f(V)$ and $g(V)$ are the mappings in the SDE \eqref{eq:PCDSDE} and $N_t$ is a Poisson process.
\end{lemma}

We first analyze the change of $\Ep{V^2(x(t))}$ during a replacement.

\begin{lemma}\label{lemma:jumps_V2}
During the replacement of an agent, the aggregate function $V(x(t))$ defined in \eqref{eq:aggregate_function_norm} satisfies:
\begin{align*}
  \Ep{V^2(x(t^+))}\!-\!\Ep{V^2(x(t^-))}\!\le\!-\frac{2}{n_0}\Ep{V^2(x(t^-))}\!+\!\frac{\Ep{\Theta^4}}{n_0^2}&\\
+\prt{\frac{2\Ep{\Theta^2}(n_0-1)+1}{n_0^2}}\Ep{V(x(t^-))},  
\end{align*}
where $\Theta$ is the new value of a replaced agent according to Assumption~\ref{ass:replacement_poisson_process}.
\end{lemma}
\begin{proof}
During the replacement of an agent $x_j$, the function $V^2$ satisfies:
\small
\begin{align*}
    (V^2)^+\!&=\prt{V^-+\frac{\Theta^2}{n_0}-\frac{(x_j^2)^-}{n_0}}^2\\
    &=(V^2)^-\!+\!\frac{\Theta^4}{n_0^2}\!+\!\frac{(x_j^4)^-}{n_0^2}\!+\!\frac{2\Theta^2V^-}{n_0} \!-\!\frac{2(x_j^2)^-V^-}{n_0}\!-\!\frac{2\Theta^2(x_j^2)^-}{n_0^2},
\end{align*}
\normalsize
where $(V^2)^+$ and $(V^2)^-$ denote $V^2(x(t^+))$ and $V^2(x(t^-))$ respectively.
Then, the conditional expectation is given by:
\small
\begin{align}
    \mathbb{E}\brk{\!(V^2)^+\!|x^-\!,\Theta}&=\frac{1}{n_0}\sum_{j=1}^{n_0}\Bigg((V^2)^-+\frac{\Theta^4}{n_0^2}+\frac{(x_j^4)^-}{n_0^2}+\frac{2\Theta^2}{n_0}V^-\nonumber\\
    &\quad-\frac{2(x_j^2)^-}{n_0}V^--\frac{2\Theta^2(x_j^2)^-}{n_0^2}\Bigg)\nonumber\\
    &\le (V^2)^-+\frac{\Theta^4}{n_0^2}+\frac{1}{n_0^3}\sum_{j=1}^{n_0}(x_j^2)^-+\frac{2\Theta^2}{n_0}V^-\nonumber\\
    &\quad-\frac{2(V^2)^-}{n_0}-\frac{2\Theta^2}{n_0^2}V^-\nonumber\\
    &=\!\prt{\frac{n_0\!-\!2}{n_0}}\!\!(V^2)^-
    \!+\!\frac{\Theta^4}{n_0^2}\!+\!\prt{\frac{1}{n_0^2}\!+\!\frac{2\Theta^2}{n_0}\!-\!\frac{2\Theta^2}{n_0^2}}\!V^-.\nonumber
\end{align}
\normalsize
Finally, we compute the total expectation to get the desired result.
\end{proof}

\begin{proposition}[Second moment]\label{prop:second_moment}
Consider a SIS epidemic satisfying the stability condition \eqref{eq:stability_condition_graph}. Under Assumptions~\ref{ass:invariant_connections} and \ref{ass:replacement_poisson_process}, the aggregate function $V(x(t))$ defined in \eqref{eq:aggregate_function_norm} satisfies:

\vspace{-6mm}

\small{
\begin{multline}\label{eq:second_moment}  \limsup_{t\to\infty}\Ep{V^2(x(t))}\!\le\! \\\frac{\mu\mathbb{E}\!\!\brk{\Theta^4}\!\!(\mu\!+\!2n_0 \bar\delta\!-\!2 \bar\beta p(n_0\!-\!1))\!+\!\mu^2\mathbb{E}\!\!\brk{\Theta^2}\!\!(2\mathbb{E}\!\!\brk{\Theta^2}\!\!(n_0\!-\!1)\!+\!1)}{2n_0(\mu+2n_0 \bar\delta-2 \bar\beta p(n_0-1))^2}.
\end{multline}
}
\normalsize
\end{proposition}
\begin{proof}
Thanks to It\^o Lemma with $\phi(V)=V^2$ we obtain:
\begin{align}
dV^2&\le \inProd{2V}{2(\beta\bar\lambda_1-\delta)V}dt+\prt{(V^2)^+-(V^2)^-}dN_t\nonumber\\
&=4(\beta\bar\lambda_1-\delta)V^2dt+\prt{(V^2)^+-(V^2)^-}dN_t.\label{eq:SDE_inequality_second_moment}
\end{align}
We take the expectation in \eqref{eq:SDE_inequality_second_moment} and we get:
\begin{align*}
d\Ep{V^2}\!\!&\le\! 4(\beta\bar\lambda_1\!-\!\delta)\Ep{V^2}\!dt\!+\!\prt{\Ep{(V^2)^+}\!-\!\Ep{(V^2)^-}}\!\mu dt\\
&=2\prt{\frac{2n_0 \beta\bar\lambda_1-2n_0 \delta-\mu}{n_0}}\Ep{V^2}dt\\
&\quad \!\!+\mu\!\!\prt{\frac{2\Ep{\Theta^2}\!(n_0-1)+1}{n_0^2}}\!\Ep{V}dt+\frac{\mu\Ep{\Theta^4}}{n_0^2} dt,
\end{align*}
where we use the Dominated Convergence Theorem to interchange the expectation and the derivative of $V^2$ and we apply Lemma~\ref{lemma:jumps_V2}. Let us denote \linebreak $a_{V}\!=\!2\prt{\frac{2n_0 \beta\bar\lambda_1-2n_0 \delta-\mu}{n_0}}$, $b_{V}\!=\!\mu\prt{\frac{2\Ep{\Theta^2}(n_0-1)+1}{n_0^2}}$ and $c_{V}\!=\!\frac{\mu\Ep{\Theta^4}}{n_0^2}$, so that we have
\begin{equation}\label{eq:ODE_second moment}
\frac{d}{dt}\Ep{V^2}\le a_{V}\Ep{V^2}+b_{V}\Ep{V}+c_{V}.    
\end{equation}
The solution of \eqref{eq:ODE_second moment} is given by:
\begin{align*}
\Ep{V^2}\!&\le e^{a_{V}t}\Ep{V_0^2}+e^{a_{V}t}\int_0^te^{-a_{V}\tau}\prt{b_{V}\Ep{V(x(\tau))}+c_{V}} d\tau  \nonumber\\
&=e^{a_{V}t}\Ep{V_0^2}\!+b_Ve^{a_{V}t}\!\!\!\int_0^t\!\!e^{-a_{V}\tau}\Ep{V}d\tau\!-\!\frac{c_{V}}{a_V}\!\prt{1\!-\!e^{a_Vt}}.
\end{align*}
We recall that by \eqref{eq:solution_ODE_SIS}, the function $\Ep{V}$ satisfies:
$$
\Ep{V(x(t))}\le d_V e^{\frac{a_V}{2}t}+m_V,
$$
where $d_V\!=\!\Ep{V_0}-\frac{\mu(\sigma^2+m^2)}{\mu+2n_0 \delta-2n_0 \beta\bar\lambda_1}$ and  $m_V\!=\!\frac{\mu(\sigma^2+m^2)}{\mu+2n_0 \delta-2n_0 \beta\bar\lambda_1}$. Then $\Ep{V^2}$ satisfies:
\begin{align}
    \Ep{V^2}&\le e^{a_{V}t}\Ep{V_0^2}-\frac{c_{V}}{a_V}\prt{1-e^{a_Vt}}\nonumber\\
    &\quad+b_Ve^{a_{V}t}\int_0^te^{-a_{V}\tau}\prt{d_V e^{\frac{a_V}{2}\tau}+m_V}d\tau\nonumber\\
    &=\prt{\Ep{V_0^2}+\frac{c_V}{a_V}+\frac{2b_Vd_V}{a_V}+\frac{b_Vm_V}{a_V}}e^{a_Vt}\nonumber\\
    &\quad-\frac{2b_Vd_V}{a_V}e^{\frac{a_V}{2}t}-\prt{\frac{c_V}{a_V}+\frac{b_Vm_V}{a_V}}.\label{eq:solution_ODE_SIS_second_moment}
\end{align}
Finally, since $a_V$ is negative, the result follows by taking the limit $t\to\infty$ in \eqref{eq:solution_ODE_SIS_second_moment} and using $\bar\lambda_1=p(n_0-1)$, $\delta=\bar\delta$ and $\beta=n_0^{-1}\bar\beta$.
\end{proof}

The bound \eqref{eq:second_moment} is derived using \eqref{eq:SIS_limit}, which corresponds to the asymptotic behavior of $\Ep{V(x(t))}$. Unlike the previous result \eqref{eq:SIS_limit}, the value of \eqref{eq:second_moment} is also affected by the fourth moment of $\Theta$, while \eqref{eq:SIS_limit} only depends on the second moment ($\Ep{\Theta^2}=\sigma^2+m^2$).

The previous bound on the second moment can also be used to estimate the variance   
\begin{align*}
 \limsup_{t\to\infty} \Var(V)&=   \limsup_{t\to\infty} \prt{\Ep{V^2}-\prt{\Ep{V}}^2}\\
 &\le \limsup_{t\to\infty} \Ep{V^2},
\end{align*}
although the result is likely to be conservative since  the effect of  $-\prt{\Ep{V(x(t))}}^2$ is neglected.

\section{Numerical and Simulation Results}

To illustrate the main results of this paper, we consider a SIS epidemic over a 
graph with $n_0=50$ agents generated with a probability $p=0.5$. The parameters of the epidemic are chosen as $\beta_n=\bar\beta/n=0.1/n$ and $\delta_n=\bar\delta=1.5p\bar\beta$, and the Poisson processes for arrivals, departures and replacements have a rate $\mu_a=\mu_d=\mu=7$. The states of the arriving agents are generated with a uniform distribution with $m=1/2$ and $\sigma^2=1/12$. Fig.~\ref{fig:star_graph_simple} presents a realization of the stochastic process. In the top plot, we show the evolution of $V(x(t))$ where the aggregate function jumps during the occurrence of arrivals (dashed green line) and departures (dash-dotted red line). In the bottom plot, we can see the changes of the number of agents during the evolution of the epidemic, which remains around the initial number $n_0=50$.

\begin{figure}
\centering
{ \includegraphics[width=\linewidth]{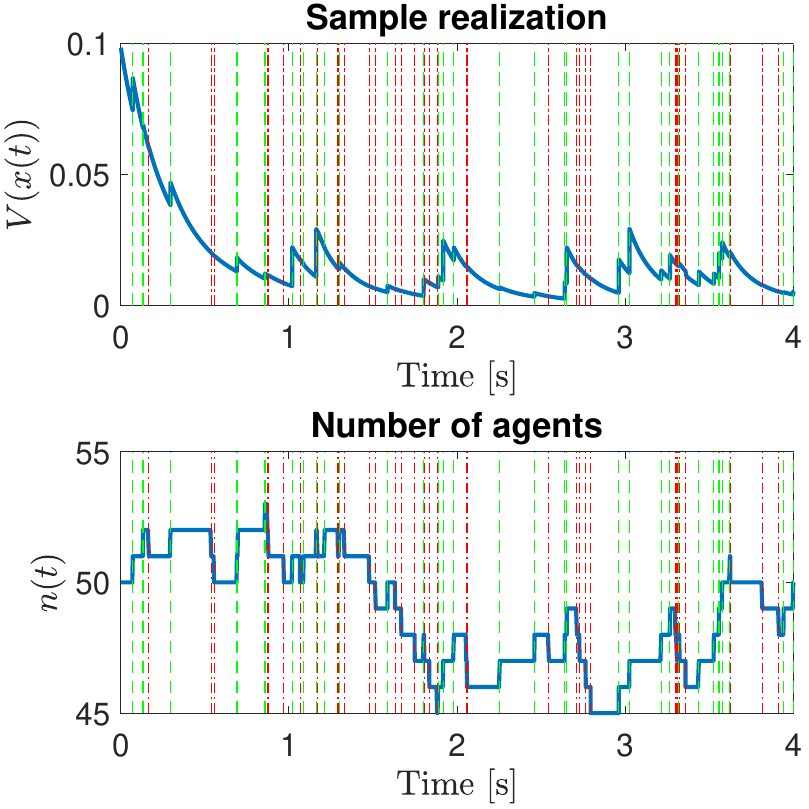}
}
\vspace{-0.6cm}
\caption{Sample realization of the aggregate function $V(x(t))$ 
and the number of agents $n(t)$ for a graph with $n_0=50$ agents and probability $p=0.5$. The SIS epidemic is considered with $\beta_n=0.1/n$ and $\delta_n=0.15p$, and the rate of the Poisson processes is $\mu_a=\mu_d=\mu=7$. The arrival and departure instants correspond to the green and red lines respectively.}\label{fig:star_graph_simple}
\vspace{-0.4cm}
\end{figure}

In Fig.~\ref{fig:moments_linear_model_SIS}, we present the computations of the moments of $V(x(t))$. 
The top plot corresponds to the evolution of $\Ep{V(x(t))}$ where the solid blue line is the original process, the dashed red line is the approximate process and the dash-dotted yellow line corresponds to the upper bound \eqref{eq:SIS_limit}.
Similarly, the middle plot shows the evolution of $\Ep{V^2(x(t))}$ and the bottom plot presents the evolution of Var$(V(x(t)))$.
Notice that the behavior of $\Ep{V(x(t))}$, $\Ep{V^2(x(t))}$ and $\text{Var}(V(x(t)))$ of the process with arrivals/departures and the process with replacements are really close. 
\begin{figure}[!ht]
\centering
{ \includegraphics[width=\linewidth]{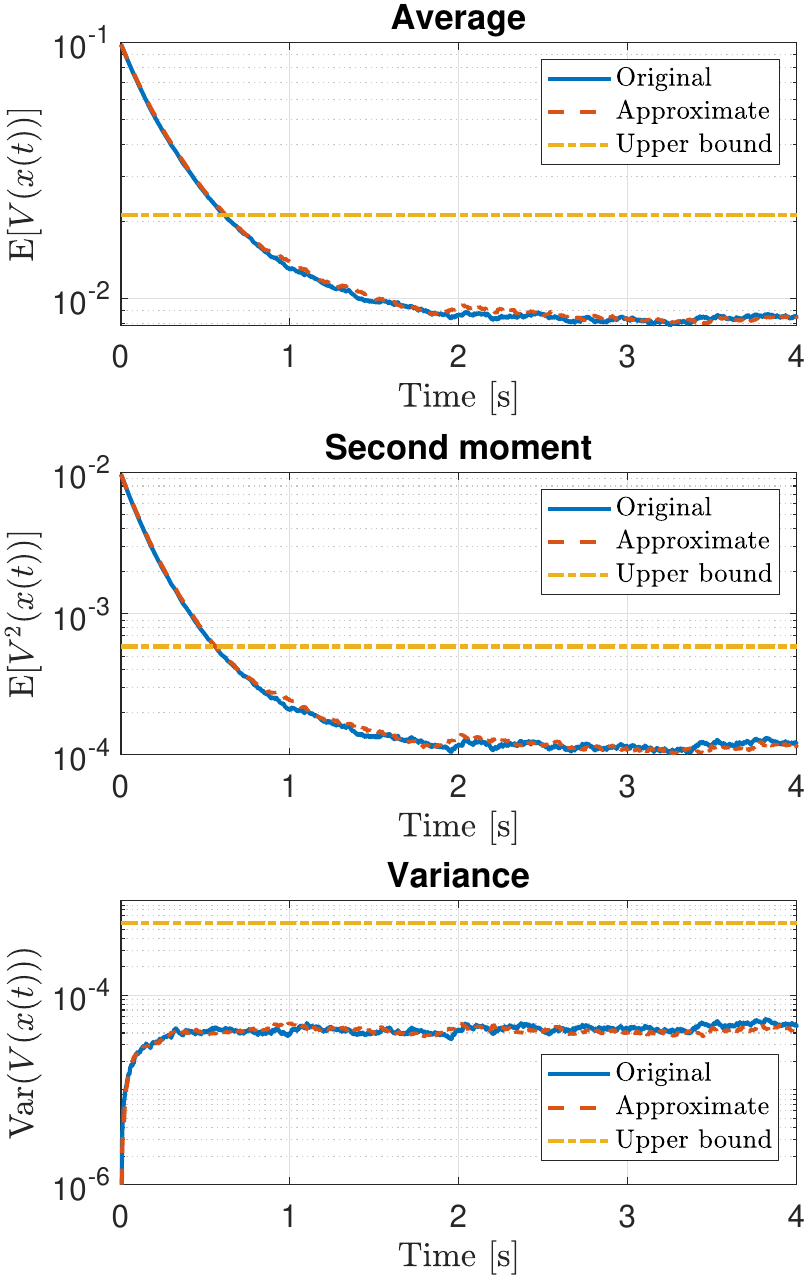}\label{fig:second_moment_V}
}
\vspace{-0.6cm}
\caption{Evolution of the moments of the aggregate function $V(x(t))$ 
for a graph with $n_0=50$ agents and probability $p=0.5$. The SIS epidemic is considered with $\beta=0.1/n$ and $\delta=0.15p$, and the rate of the Poisson processes is $\mu_a=\mu_d=\mu=7$. In the top plot, the solid blue line and dashed red line correspond to the estimation of $\Ep{V(x(t))}$ for the original and approximate process respectively, while the dash-dotted yellow line is the upper bound \eqref{eq:SIS_limit}.
In the middle plot, the solid blue line and dashed red line correspond to the estimation of $\Ep{V^2(x(t))}$ for the original and approximate process respectively, while the dash-dotted yellow line is the upper bound \eqref{eq:second_moment}. In the bottom plot, the solid blue line and dashed red line correspond to the estimation of $\Var(V(x(t)))$ for the original and approximate process respectively, while the dash-dotted yellow line is the upper bound \eqref{eq:second_moment}. 
The simulated values were obtained considering 1000 realizations of the process. 
}\label{fig:moments_linear_model_SIS}
\vspace{-0.4cm}
\end{figure}

\section{Conclusion}

In this paper, we analyzed a SIS epidemic in open multi-agent systems by using replacements to approximate a process with similar rates of arrivals and departures, and a fixed network to approximate time-varying graphs that are sampled from the same distribution.
We defined an aggregate function to analyze the epidemic and derived upper bounds for the asymptotic values of the expectation and variance as functions of the parameters of the infection and replacement processes. Through simulations we showed that the original process (i.e., with arrivals and departures) and the approximate process (i.e., with replacements) are really close, such that replacements can be used as a tool to analyze real open systems.

For future work, we would like to consider the more general and challenging case when arrivals and departures have different rates \cite{monnoyer2020open}, which would not preserve the initial number of agents in expectation. 
Finally, it would be important to extend the results to process where the connections of new agents are given by a more complex distribution. In this direction, graphons, which have already been used to study epidemics in large networks~\cite{vizuete2020graphon}, appear to be a promising model to generate network topologies that are consistent across time \cite{vizuete2022contributions}.

\bibliographystyle{IEEEtran}
\bibliography{SIS.bib}

\end{document}